\documentclass[10pt]{article}
\usepackage[colorlinks = true, linkcolor=blue, citecolor=blue, pdfpagemode = none, pdfstartview = FitH]{hyperref}

\usepackage{fullpage}
\usepackage{amsmath}
\usepackage{amsfonts}
\usepackage{amssymb}
\usepackage{amsthm}
\usepackage{times}
\usepackage{dsfont}
\usepackage{bbm}
\usepackage[normalem]{ulem}
\usepackage{color}
\usepackage[latin1]{inputenc}
\usepackage{cleveref}
\usepackage{algorithm}
\usepackage{algpseudocode}
\usepackage[final]{graphicx}
\usepackage{authblk}

\def\be{\begin{eqnarray}}

\def\ee{\end{eqnarray}}

\def\bee{\begin{eqnarray*}}

\def\eee{\end{eqnarray*}}

\algblockx[forall]{Forall}{EndForall}%
  [1]{\textbf{for all} #1 \textbf{do}}{\textbf{end}}
\algblockx[loop]{Loop}{EndLoop}%
  [1]{\textbf{loop} #1}{\textbf{end loop}}

\newcommand{\hide}[1]{}
\newcommand{\norm}[1]{\|#1\|}

\def \bi   {\begin{itemize}}
\def \ei   {\end{itemize}}

\def \id   {\mathbbm{1}}

\newcommand{\ket}[1]{|\hspace{0.5pt}#1\hspace{0.5pt}\rangle}

\newcommand{\bra}[1]{\langle\hspace{0.5pt}#1\hspace{0.5pt}|}
\newcommand{\bracket}[2]{\langle\hspace{0.5pt}#1\hspace{0.5pt}|\hspace{0.5pt}
	#2\hspace{0.5pt}\rangle}

\newcounter{defctr}
\newcounter{examplectr}
\newtheorem{theorem}{Theorem}
\newtheorem{lemma}[theorem]{Lemma}

\newtheorem{definition}[defctr]{Definition}

\newtheorem{example}[examplectr]{Example}

\begin{document}

\title{Simulating Quantum Circuits with Sparse Output Distributions}
\author[1]{Martin Schwarz\thanks{m.schwarz@univie.ac.at}}
\author[2]{Maarten Van den Nest\thanks{maarten.vandennest@mpq.mpg.de}}
\affil[1]{\normalsize Vienna Center for Quantum Science and Technology, Faculty of Physics, University of Vienna, Austria}
\affil[2]{\normalsize Max Planck Institut f\"{u}r Quantenoptik, Hans-Kopfermann-Str. 1, D-85748 Garching, Germany}

\renewcommand\Authands{\hspace{1cm}}
\date{}

\maketitle

\begin{abstract}
We show that several quantum circuit families can be simulated efficiently classically if it is promised that their output distribution is approximately sparse i.e. the distribution is close to one where only a polynomially small, a priori unknown subset of the measurement probabilities are nonzero. Classical simulations are thereby obtained for quantum circuits which---without the additional sparsity promise---are considered hard to simulate. Our results apply in particular to a family of Fourier sampling circuits (which have structural similarities to Shor's factoring algorithm) but also to several other circuit families, such as IQP circuits. Our results provide examples of quantum circuits that cannot achieve exponential speed-ups due to the presence of too much destructive interference i.e. too many cancelations of amplitudes. The crux of our classical simulation is an efficient algorithm for approximating the significant Fourier coefficients of a class of states called computationally tractable states. The latter result may have applications beyond the scope of this work. In the proof we employ and extend sparse approximation techniques, in particular the Kushilevitz-Mansour algorithm, in combination with probabilistic simulation methods for quantum circuits.
\end{abstract}

\bigskip

\section{Introduction}\label{intro}
In this paper we present classical algorithms for the simulation of several related classes of quantum circuits containing blocks of Quantum Fourier Transforms (QFTs). In particular, we consider $n$-qubit circuits with a QFT-Toffoli-QFT$^{-1}$ block structure followed by a (partial)  measurement immediately after the final QFT. Circuits of this kind are used in various quantum algorithms, most notably Shor's factoring algorithm. Whereas the circuits considered in this paper are unlikely to have an efficient classical simulation in general, the aim of this work is to analyze under which additional conditions an efficient classical simulation becomes possible. This provides an approach to identify features which are essential in the (believed) superpolynomial speed-ups achieved by, say, the factoring algorithm. In this paper we will  in particular place restrictions on the \emph{output distribution} of the circuit. In short,  our results are as follows: given the promise that the output distribution is \emph{approximately sparse} (or ``\emph{peaked}''`)---in the sense that only $O(\mbox{poly}(n))$ of the $O(2^n)$ probabilities have significant magnitude of $\Omega(1/\mbox{poly}(n))$---then an efficient classical simulation algorithm is provided.
Not unexpectedly, Shor's algorithm does not satisfy such sparseness promise i.e. its output distribution is ``superpolynomially flat''. Our results thus imply that the approximate sparseness promise alone suffices to bring down the (believed) superpolynomial speed-up achieved by the factoring algorithm to the realm of a classically simulatable quantum computation. Below we provide a discussion of how our findings shed light on the factoring algorithm (see \Cref{sect_statement}).

The implications of our results are twofold. First, they pose restrictions on the design of fast quantum algorithms. For example, our results show that any \emph{exact} quantum algorithm adopting the QFT-Toffoli-QFT$^{-1}$ block structure (or more generally the structures considered in Theorems \ref{thm_main1}-\ref{thm_main4}) which has as its output state a single computational basis state containing the answer of the problem, can never achieve an exponential quantum speed-up. Given the generality of the class of circuits considered, we believe that these classical simulation results may provide useful insights for the quantum algorithms community. Second, the present results have conceptual implications as follows: the exponential speed-up found in quantum algorithms is often related to the availability of interference of probability amplitudes in this model. Indeed in several quantum algorithms, first a superposition of states is created using a QFT, then amplitudes are manipulated in some nontrivial way using reversible (classical) gates, such that in a final QFT, by means of interference, only desired basis states survive whereas the amplitudes for undesired states cancel out. Our results imply that this qualitative picture has to be refined, since too much cancelation leading to only a few classical output states (let alone a single one!) can in fact be simulated efficiently classically, and thus cannot offer exponential speed-up. Indeed, our results imply that the final probability distribution must \emph{necessarily have super-polynomially large support} (e.g. in the same order as the full state space), in order to allow for exponential speed-up. Finally, since only polynomially many measurements can be performed efficiently on the output state---and thus only a small fraction of the necessarily large number of states can be sampled---the output distribution must have a special structure such that meaningful information can be recovered from just a few measurements. Notably, the coset state produced by Shor's algorithm (and its generalizations) has group structure which is indeed exploited in the classical post-processing step to recover the entire state space from just a few measurements (cf. \Cref{sect_statement}).

The proof techniques we use to obtain our results are twofold. First, we use randomized classical simulation methods for Computationally Tractable (CT) states as developed in \cite{vdN11}. Furthermore the latter methods are combined with algorithms for sublinear sparse Fourier transforms (SFTs), which have been pioneered in seminal work by Goldreich-Levin \cite{GL89} and Kushilevitz-Mansour \cite{KM91} and which have been refined throughout the last two decades \cite{mansour95,GGIMS02,AGS03,GMS05,AGGM06,iwen10,akavia10,HIKP12,HIKP12b}. Our work also provides further extensions of the above sparse approximation techniques.

Whereas to our knowledge this is the first paper which analyzes the effect of  (approximate) sparseness of the output distribution on the classical simulability of quantum circuits, from a more general point of view several works are related to the present paper (e.g. in terms of the class of quantum circuits considered or in terms of the techniques used). For example, a relevant series of papers regards \cite{YS07, ALM06, brown07}, that all focus on efficient classical simulation of the QFT with the aim of understanding better the workings of Shor's factoring algorithm. In the latter context, see also \cite{vdN12, BVN12} for classical simulations of a class of circuits involving QFTs over finite abelian groups supplemented with a particular family of group-theoretic operations (Normalizer circuits).
Classical simulation of CT states were considered in \cite{vdN11} by one of us. In the latter work,  the algorithms from Goldreich-Levin \cite{GL89} and Kushilevitz-Mansour \cite{KM91} were applied in the context of classical simulation, albeit in a rather different context compared to the present paper, namely to analyze the role of the classical postprocessing for quantum speed-ups (more particularly in Simon's algorithm).  Further work on CT states is done in \cite{stahlke13}; the latter work also analyzes the role of interference effects in quantum speed-ups (although from a different perspective then the present paper).  Below we will also make statements about classical simulability of IQP (Instanteneous Quantum Polynomial-time) circuits. In \cite{BJS11} it was shown (roughly speaking) that general IQP circuits cannot be simulated efficiently, unless the polynomial hierarchy collapses. In contrast, here we show that IQP circuits with an additional sparseness promise on the output distribution, are efficiently simulable classically. Finally, in \cite{MO10} the authors consider and generalize prior work on SFTs in a different direction i.e. unrelated to classical simulation issues; they prove a quantum Goldreich-Levin theorem and use it for efficient quantum state tomography for quantum states that are approximately sparse in the Pauli product operator basis.

\section{Main results: statements and discussion}\label{sect_statement}

We prove four theorems, all similar in spirit,  about  efficient classical simulability of classes of quantum circuits with a promise on the (approximate) sparseness of the output distributions and/or the output states.
We call a probability distribution over $2^n$ events \emph{$t$-sparse}, if only $t$ probabilities are nonzero, and \emph{$\varepsilon$-approximately $t$-sparse} if the probability distribution is $\varepsilon$-close in $\ell_1$-distance to a $t$-sparse one.
Throughout this paper we will work with qubit systems and sometimes indicate where generalizations of definitions and results to $d$-level systems are possible. The computational basis states of an $n$-qubit system are denoted by $|x\rangle$ where $x=x_1\cdots x_n$ is an bit string. The set of $n$-bit strings will be denoted by $B_n$.

A key concept we build upon in this work are \emph{computationally tractable} states introduced in \cite{vdN11}, which capture two key properties of simulable quantum states:
\begin{definition}[Computationally Tractable (CT) states] \label{def:ct}
An $n$-qubit state $\ket{\psi}$ is called `computationally tractable' (CT) if the following conditions hold:
\begin{enumerate}
\item it is possible to sample in poly$(n)$ time with classical means from the probability distribution ${\cal P}=\{p_x: x\in B_n\}$ defined by \label{def:ctsample}
$p_x=|\bracket{x}{\psi}|^2$, and
\item upon input of any bit string $x$, the coefficient $\bracket{x}{\psi}$ can be computed in $poly(n)$ time on a classical computer. \label{def:cteval}
\end{enumerate}
\end{definition}
The definition of CT states is straightforwardly generalized to states of systems of qudits. Several important state families are CT: matrix product states with polynomial bond dimension, states generated by poly-size Clifford circuits, states generated by poly-size nearest-neighbor matchgate circuits, states generated by bounded tree-width circuits (where all aforementioned circuits act on standard basis inputs). For definitions of these classes and proofs that they are CT states, we refer to \cite{vdN11}. Further examples of CT states are states generated by normalizer circuits over finite Abelian groups (acting on coset states) \cite{vdN12, BVN12}.
\begin{example} \label{example1}
For our purposes it will be especially useful to point out that the following classes of states are CT \cite{vdN11}.
\begin{itemize}
\item[(i)]
Let $|x\rangle$ be an arbitrary $n$-qubit computational basis state, let ${\cal F}$ denote the quantum Fourier transform over $\mathbbm{Z}_{2^k}$ for some $k\leq n$ (acting on any subset of $k$ qubits) and let ${\cal T}$ be a poly-size circuit of classical reversible gates (e.g. Toffoli gates), then the state ${\cal T}{\cal F}|x\rangle$ is CT.
\item[(ii)]
Let $f:B_n\to \{1, -1\}$ be a classically  efficiently computable function, then the state
$|\psi_f\rangle = \frac{1}{\sqrt{2^n}} \sum {f(x)}|x\rangle$, where the sum is over all $n$-bit strings $x$, is CT.
\end{itemize}
\end{example}
One may also consider a notion of CT states in the presence of oracles (see also \cite{BH13}). We say that an $n$-qubit state $|\psi\rangle$ is $f$-CT given access to an oracle $f:\{0, 1\}^m\to \{0 ,1\}$ (with $m=$ poly$(n)$) if conditions (a)-(b) in \Cref{def:ct} hold when allowing, instead of poly-time classical computations, poly-many queries to the oracle. For example, if  the function $f$ in (ii) is given as an oracle, the state $|\psi_f\rangle$ in \Cref{example1} is trivially $f$-CT.

Based on these definitions, we are now ready to state our main results.
\subsection*{Sparse output distributions}

\begin{theorem}\label{thm_main1}
Consider a unitary $n$-qubit quantum circuit composed of two blocks ${\cal C}= U_2U_1$ with input state $|\psi_{\mbox{\scriptsize{in}}}\rangle$. Suppose that the following conditions are fulfilled:
\begin{itemize}
\item[(a)] the state $U_1|\psi_{\mbox{\scriptsize{in}}}\rangle$ obtained after applying the first block is CT;
\item[(b)] the second block $U_2$ is a QFT (or QFT$^{-1}$) modulo $2^k$, for some $k\leq n$, applied to any subset $S$ of $k$ qubits. The circuit is followed by a measurement of the qubits in $S$ in the computational basis, giving rise to a probability distribution ${\cal P}$.
\item[(c)]   The distribution ${\cal P}$ is promised to be $\varepsilon$-approximately $t$-sparse for some $\varepsilon\leq 1/6$ and for some $t$ (and otherwise no information about ${\cal P}$ is available).
\end{itemize}
Then there exists a randomized classical algorithm with runtime poly$(n, t, 1/\varepsilon, \log{\frac{1}{\delta}})$ which outputs (by means of listing all nonzero probabilities) an $s$-sparse probability distribution ${\cal P}'$ where $s=O(t/\varepsilon)$; with probability at least $1-\delta$, the distribution ${\cal P}'$ is $O(\varepsilon)$-close to ${\cal P}$.  Furthermore, it is possible to sample ${\cal P}'$ on a classical computer in poly$(n, t, 1/\varepsilon)$ time.
\end{theorem}
Thus, if the sparseness $t$ is at most polynomially large in $n$,  if the error $\varepsilon$ is at worst polynomially small in $n$, and if $\delta= 2^{-\mbox{\scriptsize{poly(n)}}}$, then the classical simulation is efficient i.e. it runs in poly$(n)$ time, and the probability of failure is exponentially small.

We emphasize that, apart from the promise (c), no information about the structure of ${\cal P}$ is a priori available. For example, suppose that ${\cal P}$ is promised to be approximately 1-sparse, where a distribution is 1-sparse if there exists a single bit string $x^*$ which occurs with probability 1 and all other bit strings have probability 0. Then, crucially, we do not assume  knowledge of the bit string $x^*$, i.e a priori all (potentially exponentially many in $n$!) bit strings are equally likely. Perhaps surprisingly, \Cref{thm_main1} implies that a good approximation of ${\cal P}$ can nevertheless be efficiently computed.

\begin{figure}
  \begin{center}
    \includegraphics[scale=1]{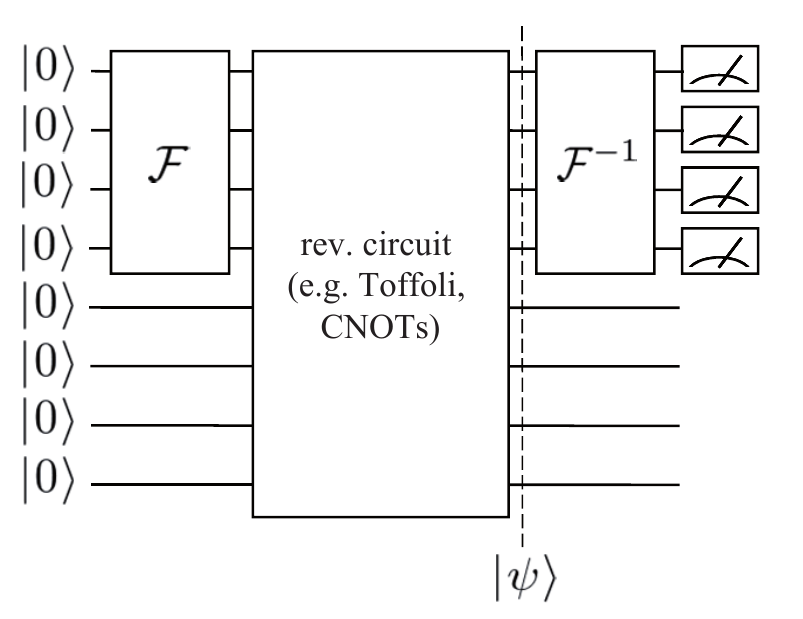}
    \caption{\emph{Shor's algorithm \cite{shor99} consists of (1) a quantum Fourier transform (QFT) on a subset of qubits, (2) a block of reversible gates (a
		modular exponentiation circuit), and (3) an inverse QFT on the same subset qubits.
		Note that the state $|\psi\rangle$ obtained after the first QFT is a computationally tractable (CT) state. Thus conditions (a) and (b) of \Cref{thm_main1} are satisfied. However
		the output distribution of Shor's algorithm is \emph{not} sparse in general, as required by our algorithm (cf. condition (c)).}}
    \label{fig:ftf-circuit}
  \end{center}
\end{figure}

Since several circuit families satisfy condition (a) (recall examples above and see \cite{vdN11}), \Cref{thm_main1} yields an efficient classical simulation of various types of circuits. For example, letting $|\psi_{\mbox{\scriptsize{in}}}\rangle$ be an arbitrary computational basis input, the block $U_1$ may be e.g. any poly-size Clifford circuit, nearest-neighbor matchgate circuit or bounded-treewidth circuit. A particularly interesting class of circuits, denoted by  ${\cal A}_{\mbox{\scriptsize{Shor}}}$, is depicted in \Cref{fig:ftf-circuit}. Note that Shor's factoring algorithm belongs to the class ${\cal A}_{\mbox{\scriptsize{Shor}}}$. It is easily verified that, for any ${\cal A}_{\mbox{\scriptsize{Shor}}}$ circuit, the state of the quantum register immediately before the second QFT is CT (recall \Cref{example1} (i) from above). Thus any circuit in ${\cal A}_{\mbox{\scriptsize{Shor}}}$ \emph{which, in addition, satisfies the sparseness condition (c) of \Cref{thm_main1}} can be simulated efficiently classically. Upon closer inspection of Shor's factoring algorithm, one finds that its output distribution ${\cal P}_{\mbox{\scriptsize{Shor}}}$ generally contains super-polynomially many nonzero probabilities and thus (non-surprisingly) \Cref{thm_main1} does not yield an efficient classical simulation of the factoring algorithm. More precisely, the size of the support of the flat distribution ${\cal P}_{\mbox{\scriptsize{Shor}}}$ equals the multiplicative order $r$ of a randomly chosen integer $x$ modulo $N$. For a general integer $N$, the order is conjectured to be $\Omega(N/\log(N))$ on average over all $N$ \cite{arnold05,KP13}. In the case of RSA, with $N=pq$, the primes $p$ and $q$ might be \emph{chosen} such that w.h.p. $r \approx N/4$ \cite{shorCST11}. Nevertheless it is interesting that the mere promise of (approximate) sparsity of the output distribution suffices to arrive at an efficient classical simulation for all ${\cal A}_{\mbox{\scriptsize{Shor}}}$ circuits, without otherwise restricting the allowed operations. This implies that the feature that ${\cal P}_{\mbox{\scriptsize{Shor}}}$ is sufficiently flat is an essential ingredient in the (believed) superpolynomial speed-up achieved by Shor's factoring algorithm.

Another observation is the following. Any quantum circuit ${\cal A}$ satisfying (a)-(b) in \Cref{thm_main1} (for example any ${\cal A}_{\mbox{\scriptsize{Shor}}}$ circuit) which, when implemented on a quantum computer, aspires to deliver a superpolynomial speed-up over classical computers, must generate a distribution ${\cal P}$ which cannot be well-approximated by a poly$(n)$-sparse distribution. At the same time, at most poly$(n)$ repetitions of ${\cal A}$ are allowed if the total computational cost is to be polynomially bounded, yielding only poly$(n)$ samples of ${\cal P}$. In other words, one only has access to `few' samples of a distribution which has support on a `large' number of outputs. Yet somehow these few samples should contain sufficient information to extract the final result of the computation with high probability (working within the standard bounded-error setting). This point is nicely illustrated by considering again the factoring algorithm (or more generally the abelian hidden subgroup algorithm). Here the output distribution is (close to) the uniform distribution over an unknown \emph{group} $H$ (and determining this group is essentially the goal of the algorithm) and the final measurement only yields a small set  of $O(\log |H|)$ randomly chosen  elements of $H$. However, since such a small set of randomly generated group elements is with high probability a generating set of the group, a small number of measurements indeed suffices to determine the entire group $H$.

\Cref{thm_main1} can be extended by allowing the block $U_2$ to comprise tensor product operations, instead of the QFT:

\begin{theorem}\label{thm_main2}
The conclusions of \Cref{thm_main1} also apply if condition (b) is replaced by
\begin{itemize}
\item[(b')]  the second block $U_2$ is an arbitrary tensor product unitary operation $U_2=u_1\otimes \cdots \otimes u_n$. The circuit is followed by a measurement of an arbitrary subset of qubits $S$ in the computational basis, giving rise to a probability distribution ${\cal P}$.
\end{itemize}
In addition, the conclusions of \Cref{thm_main1} also apply when $U_2$ is a tensor product operation as in (b'), but now for quantum algorithms operating on the Hilbert space ${\cal H}=\mathbbm{C}_{d_1}\otimes \cdots\otimes \mathbbm{C}_{d_n}$ with $d_i=O(1)$ but otherwise arbitrary, i.e. ${\cal H}$ is a system of $n$ qudits of possibly different dimensions.
\end{theorem}
A first example of the setting considered in \Cref{thm_main2} regards the family of IQP circuits (Instantaneous Quantum Polynomial time \cite{SB09}). Here the input is an $n$-qubit computational basis state $|x\rangle$ and the circuit consists of gates of the form $\exp[i\theta T]$  where $\theta$ is an arbitrary real parameter and where $T$ is a tensor product of the form $T = T_1\otimes\cdots\otimes T_n$ with $T_i\in\{I, X\}$. Since $X = HZH$,  every IQP circuit ${\cal C}$ can be written as ${\cal C}= H^{\otimes n} {\cal C}'H^{\otimes n}$ where ${\cal C}'$ is obtained by replacing each gate $\exp[i\theta T]$ by $\exp[i\theta T']$ where $T'=  T_1'\otimes\cdots\otimes T_n'$ with $T_i' = HT_iH$. Thus $T'$ is a tensor product of $Z$ operators and identity gates and hence each gate $e^{i\theta T'}$ is diagonal in the computational basis. Setting $U_1:= {\cal C}'H^{\otimes n}$ and $U_2:= H^{\otimes n}$ we find that conditions (a)-(b') of \Cref{thm_main2} are fulfilled; indeed it is straightforward to show that ${\cal C}'H^{\otimes n}|x\rangle$ is a CT state. Thus \Cref{thm_main2} shows that any IQP circuit with an approximately sparse output distribution can be simulated efficiently classically. This result is particularly interesting when compared to a hardness-of-simulation result obtained for general IQP circuits (i.e. without sparseness promise) in \cite{BJS11}. In the latter work it  was shown that an efficient, approximate classical simulation of IQP circuits (w.r.t. a certain multiplicative approximation) would imply a collapse of the polynomial hierarchy.

A second example of the setting considered in  \Cref{thm_main2} is the following. Consider a finite, possibly non-abelian group $G$ given as a direct product of $n$ individual groups, $G= G_1\times \cdots \times G_n$ where the order of each $G_i$ is $O(1)$. Define a Hilbert space ${\cal H}_G$ with computational basis vectors $|g\rangle = |g_1\rangle\otimes\cdots\otimes |g_n\rangle$ labeled by group elements $g=(g_1, \dots, g_n)\in G$. The space ${\cal H}_G$ is naturally associated with a tensor product of $n$ individual spaces, each of constant dimension. We may now consider quantum circuits of the following kind. The total Hilbert space is  ${\cal H}_G\otimes {\cal H}_n$ where ${\cal H}_n$ is an $n$-qubit system. In analogy to \Cref{fig:ftf-circuit}, we consider circuits of the block structure ${\cal C}= A_3A_2A_1$ where $A_1$ is the QFT over $G$ acting on the register ${\cal H}_G$, $A_2$ is an arbitrary poly-size circuit of classical reversible gates acting on the entire system and $A_3$ is the inverse QFT over $G$. The input is $|1_G ,0_n\rangle$ where $1_G$ is the neutral element in $G$ and $0_n$ denotes the all-zeros $n$-bit string; the circuit is followed by  measurement of the system ${\cal H}_G$ in the basis $\{|g\rangle\}$. Circuits of this kind are are of interest in the context of quantum algorithms for the (non-abelian) Hidden subgroup problem (see e.g. \cite{AMR07,lomont04}). For a definition of the QFT over a finite group we refer to e.g. \cite{MRR06}; here it suffices to mention that the QFT over a product group $G= G_1\times \cdots \times G_n$ is a tensor product operator. Furthermore it is easily verified (recall also the discussion on CT states above) that condition (a) in \Cref{thm_main1} is satisfied with $U_1 \equiv A_2A_1$. Thus \Cref{thm_main2} implies that any quantum circuit of this kind which has an approximately sparse output distribution can be simulated classically. This gives an example of a quantum circuit family comprising non-abelian QFTs (albeit of a restricted kind) which can be simulated classically.
For other examples of simulations of non-Abelian QFTs we refer to \cite{BV11}.

\subsection*{Sparse output states}

Let us present two more results regarding quantum circuits of the kinds considered in \Cref{thm_main1} and \Cref{thm_main2}, when promised that the output \emph{state} is approximately sparse. In this case we show how an approximation of the latter output state can be efficiently determined by means of a classical randomized algorithm.

An $n$-qubit state $|\varphi\rangle$ is called $\varepsilon$-approximately $t$-sparsee if there exists a state $|\varphi'\rangle$ which is $\varepsilon$-close to $|\psi\rangle$ and for which at most $t$ amplitudes $\langle x|\varphi'\rangle$ (with $|x\rangle$ computational basis states) are nonzero (see also section \ref{sect_sparseness}).

\begin{theorem}\label{thm_main3}
Consider a unitary $n$-qubit quantum circuit composed of two blocks ${\cal C}= U_2U_1$ with input state $|\psi_{\mbox{\scriptsize{in}}}\rangle$. Suppose that the following conditions are fulfilled:
\begin{itemize}
\item[(a)] the state $U_1|\psi_{\mbox{\scriptsize{in}}}\rangle$ obtained after applying the first block is CT;
\item[(b)] the second block $U_2$ is the QFT modulo $2^n$ or its inverse.
\item[(c)] The final state $|\psi_{\mbox{\scriptsize{out}}}\rangle = {\cal C}|\psi_{\mbox{\scriptsize{in}}}\rangle$ is promised to be $\sqrt{\varepsilon}$-approximately $t$-sparse for some $\varepsilon\leq 1/6$ and some $t$.
\end{itemize}
Then there exists a randomized classical algorithm with runtime poly$(n, t, 1/\varepsilon, \log{\frac{1}{\delta}})$ which outputs (by means of listing all nonzero amplitudes) an $s$-sparse state $|\psi\rangle$ which, with probability at least $1-\delta$, is $O(\sqrt{\varepsilon})$-close to $|\psi_{\mbox{\scriptsize{out}}}\rangle$, where $s=O(t/\varepsilon)$.
\end{theorem}

\begin{theorem}\label{thm_main4}
The conclusions of \Cref{thm_main3} also apply if condition (b) is replaced by
\begin{itemize}
\item[(b')]  the second block $U_2$ is an arbitrary tensor product unitary operation $U_2=u_1\otimes \cdots \otimes u_n$.
\end{itemize}
In addition, the conclusions of \Cref{thm_main3} also apply when $U_2$ is a tensor product operation as in (b'), but now for quantum algorithms operating on the Hilbert space ${\cal H}=\mathbbm{C}_{d_1}\otimes \cdots\otimes \mathbbm{C}_{d_n}$ with $d_i=O(1)$ but otherwise arbitrary.
\end{theorem}
\Cref{thm_main3} and \Cref{thm_main4} are closely connected to an important result in theoretical computer science, namely the Kushilevitz-Mansour (KM) algorithm \cite{KM91}: if one has oracle access to a Boolean function $f:B_n\to \{1, -1\}$ which is promised to have an approximately sparse Fourier spectrum, it is possible to compute a sparse approximation of $f$ in polynomial time. We connect our result to Kushilevitz-Mansour by considering \Cref{thm_main4} for an $n$-qubit system where \be |\psi_{\mbox{\scriptsize{in}}}\rangle \equiv |\psi_f\rangle = \frac{1}{2^{n/2}} \sum_x f(x)|x\rangle\ee  is a CT state, $U_1\equiv I$ and $U_2\equiv H^{\otimes n}$ where $H$ is the Hadamard gate. Then \Cref{thm_main4} implies that if $ H^{\otimes n}|\psi_f\rangle$  is promised to be approximately sparse, then a sparse approximation of the latter state can be computed efficiently. This is effectively (a version of) the KM result, stated in the language of quantum computing. Similarly, \Cref{thm_main3} relates to a version of the KM result \cite{mansour95} considered for transformations of Boolean functions under the Fourier transform over $\mathbbm{Z}_{2^n}$. The proof method of the KM theorem, suitably generalized to our setting at hand, will be an important tool for us.

\subsection*{Computing significant weights}

Whereas \Cref{thm_main1,thm_main2,thm_main3,thm_main4} involve a promise about the approximate sparseness of the output distributions/states, our final result does not. The following theorem asserts that, for CT states expanded in the Fourier basis, it is possible to efficiently determine (in a suitable approximate and probabilistic sense) all Fourier coefficients which are larger than some threshold value; a similar result also holds for CT states expanded in product bases. The result is in the present paper mainly used as a technique in the proof of \Cref{thm_main3,thm_main4} (similar to the proof of Kushilevitz-Mansour). However we believe it may be of independent interest, given the broadness of the class of CT states and the frequent usage of Fourier transforms.

Let $\mathbbm{Z}_{2^n}$ denote the cyclic group of integers modulo $2^n$. Any  $n$-bit string $x$ is identified with an element of $\mathbbm{Z}_{2^n}$ via the binary expansion.
Recall that the quantum Fourier transform over $\mathbbm{Z}_{2^n}$ is the following $n$-qubit unitary operator:
\begin{equation}
{\cal F}_{2^n} = \frac{1}{\sqrt{2^n}} \sum_{x,y \in \mathbbm{Z}_{2^n}} \exp\left({\frac{2 \pi i xy}{2^n} } \right)\ket{x}\bra{y}.
\end{equation}
and the \emph{Fourier basis} is simply the orthonormal basis $\{|F_x\rangle: x\in B_n\}$ defined by $|F_x\rangle= {\cal F}_{2^n}|x\rangle$.

\begin{theorem}\label{thm_main5}
Let $|\psi\rangle$ be an $n$-qubit CT state and consider its expansion in the Fourier basis: \be |\psi\rangle = \sum \hat\psi_x |F_x\rangle.\ee There exists a randomized classical algorithm with runtime poly$(n, \frac{1}{\theta}, \log\frac{1}{\pi})$ which outputs a list $L =\{x^1, \dots, x^l\}$ where $l\leq 2/\theta$ and where each $x^i$ is an $n$-bit string such that, with probability at least $1-\pi$:
\begin{itemize}
\item[(a)] for all $y\in L$, it holds that $|\hat\psi_x|^2\geq \frac{\theta}{2}$;
\item[(b)] every $k$-bit string $x$ satisfying  $|\hat\psi_x|^2\geq \theta$ belongs to the list $L$;
\end{itemize}
\noindent Furthermore, given any $x\in B_n$, there exists a classical algorithm with runtime poly$(n, 1/\varepsilon, \log{\frac{1}{\delta}})$ which, with probability at least $1-\delta$, outputs an $\varepsilon$-approximation of $\hat\psi_x$.
Finally, the above results also holds if the Fourier basis is replaced by a product basis $\{U|x\rangle\}$ where $U=U_1\otimes \cdots\otimes U_n$ is an arbitrary tensor product unitary operator.
\end{theorem}

\section{Proof outline and organization of the paper}

In \Cref{sect_sparseness} we discuss $\varepsilon$-approximately $t$-sparse distributions and states. A key property will be \Cref{lem:largecoeffs_prob} where we show that the large probabilities contain most of the information of an approximately sparse distribution i.e. discarding the small probabilities does not introduce too much error.

It will be a key point in our proofs that the output distributions of the quantum circuits considered in \Cref{thm_main1,thm_main2,thm_main3,thm_main4}, as well as a suitable subset of their marginal distributions,  are what will be called here \emph{additively approximable}. The latter are distributions whose individual probabilities can be efficiently approximated with a randomized classical algorithm with a performance in terms of error and success probability which is similar to the one given by the Chernoff bound.  Our analysis of additively approximable distributions (\Cref{sect_chernoff} and \Cref{sect_algorithm}), which is a significant component in the proofs of our main results, will not make reference to quantum computing (the latter is done as of \Cref{sect_CT}). In \Cref{sect_chernoff}, we introduce the notion of additively approximable distributions and develop their properties.  An important feature will be established in \Cref{thm_compute_large_coeff} where we show that, for any probability distribution which is itself additively approximable and for which a designated subset of its marginals are additively approximable as well, it is possible to efficiently determine (in a suitable approximate sense) those probabilities which are larger than some given, sufficiently large, threshold value.
This lemma, in combination with \Cref{lem:largecoeffs_prob} mentioned above, will yield an efficient algorithm to (approximately) sample any $\varepsilon$-approximately $t$-sparse distribution which is additively approximable and whose marginals are as well; this algorithm is given in \Cref{sect_algorithm} (\Cref{thm_central}). The results developed in \Cref{sect_chernoff} to \Cref{sect_algorithm} will follow the general proof idea of the Kushilevitz-Mansour theorem \cite{KM91,GL89}.

In \Cref{sect_CT} we recall classical simulation properties of CT states.  Finally,  in \Cref{sect_proofs} the proofs of our main results are given: the main strategy is to show that the output distributions of the circuits considered in our main theorems, as well as their marginals, are additively approximable.

\section{Approximate sparseness}\label{sect_sparseness}

\subsection{Basic definitions}

We call a quantum state $\ket{\varphi}$ \emph{$t$-sparse} (relative to the computational basis), if at most $t$ amplitudes $\langle x|\varphi\rangle$ are nonzero.
We will use the standard $\ell_2$-norm as as the natural distance measure for two pure states. Thus we will call two quantum states $\ket{\varphi},\ket{\psi}$ \emph{$\varepsilon$-close}, if $\norm{\ket{\varphi}-\ket{\psi}}_2 \leq \varepsilon$.
We call a normalized pure state $\ket{\varphi}$ \emph{$\varepsilon$-approximately $t$-sparse} if there exists a, not necessarily normalized, $t$-sparse vector which is $\varepsilon$-close to $\ket{\varphi}$. In this paper we will mostly be interested in a sparseness $t$ which scales at most polynomially with the number of qubits $n$, and in an error $\varepsilon$ which is worst polynomially small in $n$.
Note that in the definition of approximate sparseness we allow the $t$-sparse vector to be an unnormalized state (this will be a convenient definition in our proofs). However, if $\ket{\varphi}$ is $\varepsilon$-approximately $t$-sparse and if $\varepsilon$ is sufficiently small (namely $\varepsilon\leq 0.5$), there always exists a \emph{normalized} $t$-sparse state $|\varphi'\rangle$ which is $O(\varepsilon)$-close to $|\varphi\rangle$ as well (see \Cref{sect_sparse_basic}).

Similar to sparse quantum states, we call a probability distribution ${\cal P}=\{p_x: x\in B_n\}$ on the set of $n$-bit strings $t$-sparse if at most $t$ of its probabilities $p_x$ are nonzero. The distance between two probability distributions ${\cal P}$ and ${\cal P}'$ will be measured in terms of the total variation distance, defined by \be \|{\cal P}-{\cal P}'\|_1= \sum |p_x - p'_x|.\ee We say that ${\cal P}$ is $\varepsilon$-approximately $t$-sparse if there exists a $t$-sparse vector $v= (v_x: x\in B_n)$ such that $\sum |p_x - v_x|\leq \varepsilon$. The entries $v_x$ may a priori be arbitrary complex numbers. However, similar to above,  if ${\cal P}$ is $\varepsilon$-approximately $t$-sparse and if $\varepsilon$ is sufficiently small, there always exists a \emph{normalized probability distribution} ${\cal P}'$ which is $t$-sparse and such that $\|{\cal P}- {\cal P}'\|_1\leq O(\varepsilon)$  (see \Cref{sect_sparse_basic}).

The support of a probability distribution ${\cal P}=\{p_x: x\in B_n\}$ is the set of all $x$ for which $p_x\neq 0$. If $A\subseteq B_n$, the restriction of ${\cal P}$ to $A$ is the subnormalized distribution $\{q_x: x\in B_n\}$ defined by \be\label{restriction} q_x = \left\{ \begin{array}{cl} p_x &\mbox{if } x\in A \\ 0& \mbox{otherwise.}\end{array}\right.\ee
Similarly, the support of an $n$-qubit state is the set of all $x$ for which $\langle x|\varphi\rangle\neq 0$. If $A\subseteq B_n$, the restriction of $|\varphi\rangle$ to $A$ is the subnormalized state \be\sum_{x\in A} \langle x|\varphi\rangle |x\rangle.\ee

\subsection{Basic properties}\label{sect_sparse_basic}

Let ${\cal P}= \{p_x: x\in B_n\}$ be an arbitrary probability distribution. Let $A_t\subseteq B_n$ be a subset which, roughly speaking, contains  $t$ bit strings corresponding to the $t$ largest probabilities of ${\cal P}$. More formally, $A_t$ satisfies the properties (i) $|A_t|=t$ and (ii) $p_x\geq p_y$ for all $x\in A_t$ and $y\notin A_t$. Note that there may be more than one set $A_t$ with this property (e.g. if multiple probabilities happen to be equal). For our purposes the particular choice of $A_t$ will however be irrelevant. Let ${\cal P}[t]$ denote the restriction of ${\cal P}$ to $A_t$. Note that ${\cal P}[t]$ is $t$-sparse. Furthermore it is straightforward to show that, for any $t$-sparse vector $v=(v_x: x\in B_n)$ (where the $v_x$ may be arbitrary complex numbers), one has $\|{\cal P}[t] - {\cal P}\|_1 \leq \| v - {\cal P}\|_1$ i.e. ${\cal P}[t]$ has minimal distance to ${\cal P}$ among all such  $t$-sparse $v$'s.  It follows that ${\cal P}$ is $\varepsilon$-approximately $t$-sparse iff
\begin{equation}\|{\cal P} - {\cal P}[t]\|_1\leq \varepsilon.\label{eq:P_t} \end{equation}
Next we show that, for any $\varepsilon$-approximately $t$-sparse distribution ${\cal P}$ with $\varepsilon\leq 0.5$ there always exists a $t$-sparse \emph{normalized} distribution ${\cal P}'$ which is $O(\varepsilon)$-close to ${\cal P}$. To see this, set ${\cal P}':=  {\cal P}[t]/\|{\cal P}[t]\|_1$. Owing to \cref{eq:P_t} we have \begin{equation}\label{eq:bound_P_t} 1-\varepsilon\leq \| {\cal P}[t]\|_1\leq 1.\end{equation} We then find
\be\label{eq:argument_normalize}
\|{\cal P}' - {\cal P}\|_1 &=& \frac{\| {\cal P}[t]- \|{\cal P}[t] \|_1 \cdot {\cal P}\|_1}{\|{\cal P}[t] \|_1}\leq \frac{\| {\cal P}[t]- \|{\cal P}[t] \|_1 \cdot {\cal P}\|_1}{1-\varepsilon}\nonumber\\ &\leq& \frac{\| {\cal P}[t] - {\cal P}\|_1}{1-\varepsilon} + \frac{(1- \|{\cal P}[t]\|_1)\cdot \|{\cal P}\|_1}{1-\varepsilon} \leq \frac{2\varepsilon}{1-\varepsilon}.
\ee Here in the equality we used the definition of ${\cal P}[t]$; in the first inequality we used \cref{eq:bound_P_t}; in the second inequality we used the triangle inequality; finally we  used \cref{eq:P_t} and \cref{eq:bound_P_t}. Then, if $\varepsilon \leq 0.5$, we have $\|{\cal P}' - {\cal P}\|_1\leq 4\varepsilon$.

Let $|\varphi\rangle$ be an $n$-qubit state. In analogy with above, let $A_t\subseteq B_n$ be a subset which, roughly speaking, contains $t$ bit strings corresponding to the $t$ largest amplitudes of $|\varphi\rangle$. More formally, $A_t$ satisfies (i) $|A_t|=t$ and (ii) $|\langle x|\varphi\rangle|\geq |\langle y|\varphi\rangle|$ for all $x\in A_t$ and $y\notin A_t$. Letting $|\varphi[t]\rangle$ denote the restriction of $|\varphi\rangle$ to $A_t$, it is straightforward to show that $|\varphi[t]\rangle$ has minimal $\ell_2$-distance to $|\varphi\rangle$ among all $t$-sparse vectors. It follows that $|\varphi\rangle$ is $\varepsilon$-approximately $t$-sparse iff \be\label{varphi_t} \||\varphi\rangle - |\varphi[t]\rangle\|_2\leq \varepsilon.\ee
Fully analogous to above, for any $\varepsilon$-approximately $t$-sparse state $|\varphi\rangle$ with $\varepsilon\leq 0.5$ there always exists a $t$-sparse \emph{normalized} state $|\varphi'\rangle$ which is $O(\varepsilon)$-close to $|\varphi\rangle$. The state $|\varphi'\rangle:=  |\varphi[t]\rangle/\||\varphi[t]\rangle\|_2$ does the job.

Let $|\varphi\rangle$ be an $n$-qubit pure state and let ${\cal P}$ be the probability distribution arising from measuring  all qubits of $|\varphi\rangle$ in the computational basis. We may then ask whether ${\cal P}$ is approximate sparse or whether the full state $|\varphi\rangle$ is approximately sparse, where in the former case closeness is measured w.r.t. total variation distance and in the latter case it is measured w.r.t. $\ell_2$ distance. Next we show that both notions of approximate sparseness are equivalent up to a square-root rescaling of the accuracy $\varepsilon$ (which is mostly harmless if one is ultimately interested in $\varepsilon =1/$poly$(n)$, as we will mostly be in this paper).
\begin{lemma}\label{thm_sparse_equiv}
Let $|\varphi\rangle$ be an $n$-qubit pure state and let ${\cal P}$ be the probability distribution arising from measuring  all qubits of $|\varphi\rangle$ in the computational basis. Then $|\varphi\rangle$ is $\sqrt{\varepsilon}$-approximately $t$-sparse (relative to the $\ell_2$-distance, as above) iff ${\cal P}$ is $\varepsilon$-approximately $t$-sparse (relative to the total variation distance, as above).
\end{lemma}
\begin{proof} Define $p_x = |\langle x|\varphi\rangle|^2$ for all $x$. As above, let $A_t$ be a set of $t$ $n$-bit string satisfying $p_x\geq p_y$ for all $x\in A_t$ and $y\notin A_t$. This is (trivially) equivalent to $|\langle x|\varphi\rangle|\geq |\langle y|\varphi\rangle|$ for all $x\in A_t$ and $y\notin A_t$. Let ${\cal P}[t]$ denote the restriction of ${\cal P}$ to $A_t$ and similarly $|\varphi[t]\rangle$ is the restriction of $|\varphi\rangle$ to $A_t$. Recall that $|\varphi\rangle$ is $\sqrt{\varepsilon}$-approximately $t$-sparse iff $\||\varphi\rangle - |\varphi[t]\rangle\|_2\leq \sqrt{\varepsilon}$ and that ${\cal P}$ is $\varepsilon$-approximately $t$-sparse iff $\|{\cal P} - {\cal P}[t]\|_1\leq \varepsilon$. A straightforward application of definitions now shows that $\||\varphi\rangle - |\varphi[t]\rangle\|^2 =\|{\cal P} - {\cal P}[t]\|$, since both expressions coincide with \be \sum_{x\notin A_t} |\langle x|\varphi\rangle|^2.\ee This shows that $\||\varphi\rangle - |\varphi[t]\rangle\|_2\leq \sqrt{\varepsilon}$ iff $\|{\cal P} - {\cal P}[t]\|_1\leq \varepsilon$.
\end{proof}

\subsection{Sparse distributions have large coefficients}

The next lemma shows that, for an approximately sparse probability distribution, the `small' probabilities can be ignored without introducing much error. This property will be important in the proof of our main results, in combination with \Cref{thm_compute_large_coeff} which states that the large probabilities can be efficiently computed for certain distributions. The following lemma is also closely related to \cite[Lemma 3.11]{KM91}
\begin{lemma} \label{lem:largecoeffs_prob}
Let ${\cal P}=\{p_x: x\in B_n\}$ be an $\varepsilon$-approximately $t$-sparse probability distribution. Define $B_{\varepsilon, t}$ to be the subset of all bit strings $x$ such that $p_x \geq \varepsilon/t$. Define the subnormalized distribution ${\cal Q}_{\varepsilon, t}$ to be the restriction of ${\cal P}$ to $B_{\varepsilon, t}$. Then ${\cal Q}_{\varepsilon, t}$  is $O(\varepsilon)$-close to ${\cal P}$. More precisely $\|{\cal Q}_{\varepsilon, t} - {\cal P}\|_1\leq 2\varepsilon$.
\end{lemma}
\begin{proof}
Let $A_t\subseteq B_n$ and ${\cal P}[t]$ be defined as in \Cref{sect_sparse_basic}. Recall that $\|{\cal P}[t] - {\cal P}\|_1\leq \varepsilon$ owing to the approximate sparseness of ${\cal P}$.
Furthermore construct ${\cal P}'$ as follows: start from ${\cal P}[t]$ and set all probabilities with magnitudes $\leq \frac{\varepsilon}{t}$ to zero; let $C$ denote the support of ${\cal P}'$. Note that $C\subseteq A_t$ and thus $|A_t\setminus C|\leq t$. Furthermore $p_x\leq  \varepsilon/t$ for every $x\in A_t\setminus C$. Then
\be\label{B_epsilon_t}
\norm{{\cal P}-{\cal P}'}_1 &\leq&  \norm{{\cal P}-{\cal P}[t]}_1 + \norm{{\cal P}[t]-{\cal P}'}_1 \nonumber \\
&\leq &\varepsilon + \sum_{x\in A_t\setminus C} p_x
\leq \varepsilon + t\cdot \frac{\varepsilon}{t} = 2\varepsilon.
\ee
Note also that $C\subseteq B_{\varepsilon, t}$ since $p_x \geq \varepsilon/t$ for all $x\in C$. Thus both ${\cal P}'$ and ${\cal Q}_{\varepsilon, t}$ are restrictions of ${\cal P}$, and that the support $B_{\varepsilon, t}$ of ${\cal Q}_{\varepsilon, t}$  contains the support $C$ of ${\cal P}'$. This implies that $\norm{{\cal P}-{\cal Q}_{\varepsilon, t}}_1\leq \norm{{\cal P}-{\cal P}'}_1$. Together with (\ref{B_epsilon_t}) this proves the result.
\end{proof}
An analogous result holds for approximately sparse quantum states. We do not make it explicit here since it will not be needed in our proofs of the main results.

\section{Additively approximable probability distributions}\label{sect_chernoff}

\subsection{Definition and basic properties}

The Chernoff-Hoeffding bound is a basic tool in probability theory which will be used in this work. Whereas the bound is usually stated for real-valued random variables, here we state a simple generalization to the complex-valued case, which follows from the real-valued case by bounding real and imaginary parts of independently.
\begin{lemma}[Chernoff-Hoeffding bound]\label{thm:chernoff}
Let $X_1,\dots,X_T$ be i.i.d. complex-valued random variables with $E := \mathbf{E}X_i$ and $|X_i| \leq 1$ for every $i=1,\dots,T$. Then with $T=\frac{4}{\varepsilon^2} \log(\frac{4}{\delta})$ we have
\[
\Pr\left\{\left|\frac{1}{T}\sum_{i=1}^{T} X_i - E \right| \leq \varepsilon \right\} \geq 1 - \delta
\]
\end{lemma}
\noindent A proof of \Cref{thm:chernoff} can be found in \Cref{sec:app_chernoff}. The main application of the Chernoff bound used in this work will be in the following context. Let $F:B_n\to \mathbbm{C}$ be an efficiently computable complex function (i.e. computable in polynomial time on a deterministic classical computer) satisfying $|F(x)|\leq 1$ for all $x\in B_n$ and let ${\cal P}:=\{p_x:x\in B_n\}$ be a probability distribution on the set of $n$-bit strings which can be sampled in poly$(n)$ time on a randomized classical computer. Then a direct application of the Chernoff-Hoeffding bound shows that there exists a classical randomized algorithm to estimate
\begin{equation}
\langle F\rangle:= \sum p_x F(x) \label{F}
\end{equation}
with error $\varepsilon$ and probability at least $1-\delta$ in poly$(n, \frac{1}{\varepsilon}, \log\frac{1}{\delta})$ time. This means that in poly$(n)$ time it is possible to achieve an accuracy $\varepsilon= 1/$poly$(n)$ and exponentially small failure probability $\delta= 2^{-\mbox{\scriptsize{poly}(n)}}$.

Next we introduce a definition for functions that are approximable with randomized classical algorithms having a performance in terms of error $\varepsilon$ and failure probability $\delta$ that is analogous to those obtained by applying the Chernoff bound (see also \cite{BFLW05} for a related notion of additive approximations).

\begin{definition}\label{Chernoff_comp}
A function $f:B_n\to\mathbbm{C}$ is said to be additively approximable if their exists a randomized classical algorithm with runtime poly$(n, 1/\varepsilon, \log\frac{1}{\delta})$ which, on input of an $n$-bit bit string $x$, outputs with probability at least  $1-\delta$ an $\varepsilon$-approximation of $f(x)$. A probability distribution ${\cal P}= \{p_x\}$ on the set of $n$-bit strings is said to be additively approximable if the function $x\to p_x$ is additively approximable.
\end{definition}
Note that any ${\cal P}$ which can be sampled classically in poly$(n)$ time is additively approximable since each individual probability can essentially be computed by sampling the distribution. More precisely,  to estimate $p_x$, write $p_x = \sum \delta(x, y) p_y$ where $\delta(x, y)$ equals $1$ if $x=y$ and 0 otherwise. We have thus rewritten $p_x$ as the expectation value of $F\equiv \delta (x, \cdot)$ which is a poly$(n)$-time computable function satisfying $|F(x)|\leq 1$ for all $x\in B_n$. The discussion above \Cref{Chernoff_comp} then immediately implies that ${\cal P}$ is additively approximable.

In the example discussed in \cref{F} we found that $\langle F\rangle$ can be efficiently approximated provided that $F$ was efficiently computable on a deterministic computer. In the following lemma it is shown that the same performance in estimating $\langle F\rangle$ can be achieved even when $F$ is only additively approximable. The argument is a basic application of the Chernoff bound.

\begin{lemma}\label{thm_approximate_F}
Let $F:B_n\to \mathbbm{C}$ be an additively approximable function and let ${\cal P}:=\{p_x: x\in B_n\}$ be a probability distribution  which can be sampled in poly$(n)$ time on a classical computer. Then there exists a classical randomized algorithm to estimate $\langle F\rangle:= \sum p_x F(x)$ with error $\varepsilon$ and probability $1-\delta$ in poly$(n, \frac{1}{\varepsilon}, \log\frac{1}{\delta})$ time.
\end{lemma}
\begin{proof}
By generating $K=O(\frac{1}{\varepsilon^2}\log \frac{1}{\delta})$ bit strings $x^1, \dots, x^K$ from the distribution ${\cal P}$, the inequality \be\label{chernoff_prop} \left|\frac{1}{K} \sum_{i=1}^K F(x^i) - \langle F\rangle\right|\leq \varepsilon/2\ee holds with probability at least $1-\delta/2$, owing to the Chernoff bound. Then, for each $x^i$ we compute a complex number  $c_i$ satisfying  $|c_i - F(x^i)|\leq \frac{\varepsilon}{2}$ with probability at least $1- \delta/(2K)$. Since $F$ is additively approximable, each $c_i$ can be computed in time \be\begin{array}{c} T= \mbox{poly} (n, \frac{2}{\varepsilon}, \log \frac{2K}{\delta}) = \mbox{poly} (n, \frac{1}{\varepsilon}, \log \frac{1}{\delta}).\end{array}\ee   Thus the total runtime of computing all values $c_i$ is $KT = $ poly$(n, \frac{1}{\varepsilon}, \log \frac{1}{\delta}).$ The total probability that each $c_i$ is $\frac{\varepsilon}{2}$-close to $F(x^i)$ \emph{and} that (\ref{chernoff_prop}) holds is at least
\be\begin{array}{c} (1- \frac{\delta}{2})\cdot (1-\frac{\delta}{2K})^K\geq (1- \frac{\delta}{2})\cdot (1-\frac{\delta}{2})\geq 1-\delta\end{array}\ee where we have repeatedly used that $(1-a)^r\geq 1-ra$ for all positive integers $r$ and for all $a\in[0, 1]$.   It follows that, with probability at least $1-\delta$, we have \be |\frac{1}{K} \sum_{i=1}^K c_i - \langle F\rangle|\leq \varepsilon\ee by using the triangle inequality.

\end{proof}

\subsection{Estimating large coefficients}
The following theorem contains the property of additive approximations which is most important for our purposes. It is a statement that, for distributions which are additively approximable and for which also (a designated subset of) the marginals are additively approximable, there exists an efficient algorithm to determine  those probabilities which are larger than some given threshold value. The proof technique is a type of binary search algorithm which is a direct generalization of the proof of the Kushilevitz-Mansour algorithm \cite{KM91}.
\begin{theorem}\label{thm_compute_large_coeff}
Let ${\cal P}= \{p_x: x\in B_k\}$ be a probability distribution. Let ${\cal P}_m$ denote the marginal probability distribution of the first $m$ bits, for every $m$ ranging from 1 to $k$ (with ${\cal P}_k\equiv {\cal P}$). Suppose that all distributions ${\cal P}_m$ are additively approximable. Then the following holds: given $\theta, \pi>0$, there exists a randomized classical algorithm with runtime poly$(k, \frac{1}{\theta}, \log\frac{1}{\pi})$ which outputs a list $L =\{x^1, \dots, x^l\}$ where $l\leq 2/\theta$ and where each $x^i$ is an $k$-bit string such that, with probability at least $1-\pi$:
\begin{itemize}
\item[(a)] for all $y\in L$, it holds that $p(y)\geq \frac{\theta}{2}$; \label{it:lc_a}
\item[(b)] every $k$-bit string $x$ satisfying  $p(x)\geq \theta$ belongs to the list $L$; \label{it:lc_b}
\end{itemize}

\end{theorem}
\begin{proof}
For any integer $m\leq k$ we denote by $p(x_1\cdots x_m)$ the marginal probability of the bit string $x_1\cdots x_m$. We point out the basic fact that
\begin{equation}
p(x_1\cdots x_{m-1})\geq p(x_1\cdots x_{m-1} x_m) \label{eq:basic_fact}
\end{equation} for all $m$ and for all $x_j$'s.

The algorithm will consist of $k$ steps. In each step we construct a list $L_m$ containing a certain collection of $m$-bit strings, where $m$ ranges from 1 to $k$. The final list $L_k$ will satisfy (a)-(b) with probability at least $1-\pi$. In the algorithm we will repeatedly invoke that each ${\cal P}_m$ is additively approximable; whenever an additive approximation of any ${\cal P}_m$ will be considered, we will set the required probability of success to be at least $1-\delta$ with $\delta:= \theta\pi/2k$ and the accuracy to be $\varepsilon:=\theta/4$. Each single estimate of such a probability can be done in time \be N_{\mbox{\scriptsize{single}}} = \mbox{ poly}(k, \frac{1}{\varepsilon}, \log\frac{1}{\delta}) = \mbox{ poly}(k, \frac{1}{\theta}, \log\frac{1}{\pi}).\ee

{\bf Step 1.} The list $L_1\subseteq B_1\equiv\{0, 1\}$ is computed as follows. We use that ${\cal P}_1$ is additively approximable and compute $p(0)$ (i.e. the probability of the outcome 0 on the first bit). More formally, we compute a number $c(0)$  satisfying \be | c(0) - p(0)|\leq \theta/4\ee with probability at least $1-\delta$. If $c(0)\geq 3\theta/4$ then define the bit $0$ to belong to the list $L_1$. Analogously we compute $c(1)$ as an approximation of $p(1)$ and add the bit $1$ to $L_1$ if $c(1)\geq 3\theta/4$.

{\bf Step 2.} To compute the list $L_2\subseteq B_2 \equiv \{00, 01, 10, 11\}$ we use that ${\cal P}_2$ is additively approximable as follows. For every $x\in L_1$ and $u\in\{0, 1\}$ we compute an $\theta/4$-approximation of $p(xu)$ with  probability at least $1-\delta$, yielding a number $c(xu)$ in analogy to Step 1. If $c(xu)\geq 3\theta/4$ then we add the bit pair $xu$ to the list $L_2$.

{\bf Steps 3-k.} The above procedure is continued for all $m=3\cdots k$ where in the $m$-th step we use that ${\cal P}_m$ is additively approximable. To compute the list $L_m\subseteq B_m$, for every $x_1\cdots x_{m-1}\in L_{m-1}$ and $u\in\{0, 1\}$ we compute $c(x_1\cdots x_{m-1}u)$, which is an $\theta/4$-approximation of $p(x_1\cdots x_{m-1}u)$ with probability at least $1-\delta$. If $c(x_1\cdots x_{m-1}u)\geq 3\theta/4$ then we add the bit string $x_1\cdots x_{m-1}u$ to the list $L_m$.

Finally, if at some point in the above algorithm one of the lists $L_m$ contains strictly more than $2/\theta$ elements, the algorithm is halted and all subsequent lists $L_{m+1}, \dots, L_k$ are defined to be empty. With this extra constraint, we ensure  that at most $2k/\theta$ probabilities are estimated.  It follows that the total runtime of the algorithm is \be \frac{2k}{\theta} \cdot N_{\mbox{\scriptsize{single}}}= \mbox{ poly}(k, \frac{1}{\theta}, \log\frac{1}{\pi}).\ee Furthermore, since at most $2k/\theta$ probabilities are estimated, each succeeding with probability $1-\delta$,  the probability that all estimates succeed is at least $(1 - \delta)^{\frac{2k}{\theta}}\geq 1 - \frac{2k}{\theta}\delta= 1-\pi$.

From this point on we consider the case that all estimates succeed, and claim that  in this case the list $L_k$ satisfies (a)-(b). We make the following observations. First, for every $m$ we prove property (a'):
{\it For all $x_1\cdots x_m\in L_m$ it holds that $p(x_1\cdots x_m)\geq \theta/2$.}
This is true since $c(x_1\cdots x_m)$ is an $\frac{\theta}{4}$-approximation of $p(x_1\cdots x_m)$ and since $x_1\cdots x_m$ was only added to $L_m$ if $c(x_1\cdots x_m)\geq 3\theta/4$.
Property (a') implies that the list $L_k$ satisfies (a). Furthermore, property (a') implies that every list $L_m$ contains at most $2/\theta$ bit strings (since probability distributions are normalized to sum up to 1). This shows that, as long as all estimates of the probabilities are successful, the halting procedure described above need never be applied (indeed, the latter is only incorporated in the algorithm to ensure that successive failed estimations of probabilities do not result in an (exponentially) long runtime).

Second, we argue that each  $L_m$ satisfies property (b'):
{\it If $p(x_1\cdots x_m)\geq \theta$ then $x_1\cdots x_m\in L_m$.}
To see this, we argue by induction on $m$. For $m=1$, property (b') follows immediately from the definition of $L_1$. Furthermore suppose that $y=y_1\cdots y_{m}$ satisfies $p(y)\geq \theta$. Then, using \cref{eq:basic_fact} we have $p(y_1\cdots y_{m-1})\geq \theta$ and thus, by induction, we have $y_1\cdots y_{m-1}\in L_{m-1}$. The definition of $L_m$ now immediately implies that $y_1\cdots y_{m}\in L_m$. This shows that property (b') holds for all $L_m$, so that $L_k$ satisfies (b) as desired.
\end{proof}

\section{Algorithm for additively approximable, approximately sparse distributions}\label{sect_algorithm}

We now arrive at an efficient algorithm which, on input of a probability distribution ${\cal P}$ which is promised to be approximately sparse \emph{and} which satisfies the  conditions of \Cref{thm_compute_large_coeff}, outputs an (exactly) sparse distribution ${\cal P}'$ which is close to ${\cal P}$. In addition, the distribution ${\cal P}'$ can be sampled efficiently. The proof will be obtained by combining \Cref{thm_compute_large_coeff} and \Cref{lem:largecoeffs_prob}. The argument is straightforward but somewhat tedious since some care is required in choosing suitable epsilons and deltas. We also note that Theorem \ref{thm_central} is closely related to theorem 3.11 in \cite{KM91}, which provides a randomized classical algorithm for computing representations of Boolean functions which are promised to be approximately sparse.

\begin{theorem}\label{thm_central}
Let ${\cal P}$ be a distribution on $B_k$ which satisfies the following conditions:

\begin{itemize}
\item[(i)] ${\cal P}$ is promised to be $\varepsilon$-approximately $t$-sparse, where $\varepsilon\leq 1/6$.
\item[(ii)] ${\cal P}$ and its marginals ${\cal P}_m$ ($m=1, \dots, k$) are additively approximable as in \Cref{thm_compute_large_coeff}.
\end{itemize}
Then there exists a randomized classical algorithm with runtime poly$(k, t, \frac{1}{\varepsilon}, \log \frac{1}{\delta})$ which outputs (by means of listing all nonzero probabilities)  an $s$-sparse probability distribution ${\cal P}'=\{p_x'\}$ where $s=O(t/\varepsilon)$  such that, with probability at least $1-\delta$,  ${\cal P}'$ is $O(\varepsilon)$-close to ${\cal P}$ (more precisely $\|{\cal P}-{\cal P}'\|_1\leq 12\varepsilon$). Furthermore, $p_x'\geq \varepsilon/8t$ for all $p_x'$ which are nonzero. Finally, it is possible to sample ${\cal P}'$ on a classical computer in  poly$(k, t, 1/\varepsilon)$ time.
\end{theorem}
\begin{proof} First we invoke \Cref{thm_compute_large_coeff} with  $\theta:=\varepsilon/t$ and
\be
\begin{array}{c}\label{def_pi}
\pi: = \frac{\delta}{ 2t/\varepsilon+1}.
\end{array}\ee
This yields, with probability at least $1-\pi$, a list $L$ of $k$-bit strings satisfying conditions (a)-(b), within a runtime
\be\begin{array}{c}
N_1= \mbox{poly}(k, \frac{1}{\theta}, \log \frac{1}{\pi}) = \mbox{poly}(k, t, \frac{1}{\varepsilon}, \log \frac{1}{\delta}).
\end{array}\ee
Note that $|L|\leq 2t/\varepsilon$.  Second, since ${\cal P}$ is additively approximable, each individual probability $p_x$ with $x\in L$ can be computed with success probability at least $1-\pi$ and with an error $\varepsilon'$ set to \be\varepsilon':=\min\{ \varepsilon/|L|, \varepsilon/4t \}\ee in time
\be\begin{array}{c}
N_2= \mbox{poly}(k, \frac{1}{\varepsilon'}, \log \frac{1}{\pi}) = \mbox{poly}(k, t, \frac{1}{\varepsilon}, \log \frac{1}{\delta}).
\end{array}\ee
This yields a list of numbers $\{c_x: x\in L\}$ such that $|p_x-c_x|\leq \varepsilon'$ for all $x\in L$ if all evaluations were successful. Up to this point, the runtime of the algorithm is $N= N_1 + |L|N_2$ which scales as  poly$(k, t, \frac{1}{\varepsilon}, \log \frac{1}{\delta})$, and the total success probability is at least \be (1-\pi)^{|L|+1}\geq 1- (|L|+1)\pi\geq 1-\delta\ee where we have used (\ref{def_pi}) and the property $|L|\leq 2t/\varepsilon$. From this point on, the entire algorithm proceeds deterministically.

Define $c_x$ to be 0 for all $x\notin L$ and let ${\cal C}= \{c_x: x\in B_k\}$ denote the resulting list of $2^k$ coefficients. Now let ${\cal Q}_{\varepsilon, t}=\{q_x\}$ be the restriction of ${\cal P}$ to $B_{\varepsilon, t}$, where $B_{\varepsilon, t}$ is the set of strings satisfying $p_x\geq\varepsilon/t$, as defined in \Cref{lem:largecoeffs_prob}.  Note that $B_{\varepsilon, t}\subseteq L$ (recall condition (b) of \Cref{thm_compute_large_coeff} and the fact that here $\theta= \varepsilon/t$).  Then \be \|{\cal C} - {\cal P}\|_1 &=& \sum_{x\in L} |c_x - p_x| + \sum_{x\notin L} p_x \leq |L| \cdot\varepsilon' + \sum_{x\notin L} p_x \nonumber \\ &\leq & \varepsilon + \sum_{x\notin L} p_x \leq \varepsilon + \sum_{x\notin B_{\varepsilon, t}} p_x = \varepsilon + \|{\cal P} - {\cal Q}_{\varepsilon, t}\|_1\leq 3\varepsilon.\ee
Here in the first inequality we used that $|c_x-p_x|\leq \varepsilon'$ for all $x\in L$; in the second, we used the definition of $\varepsilon'$; in the third, we used $B_{\varepsilon, t}\subseteq L$; in the equality, we used the definition of ${\cal Q}_{\varepsilon, t}$; finally, we used \Cref{lem:largecoeffs_prob}.

Since $|c_x - p_x|\leq \varepsilon'\leq \varepsilon/4t$ (recall the definition of $\varepsilon')$ and since $p_x\geq \varepsilon/2t$ owing to condition (a) of \Cref{thm_compute_large_coeff}, we have $c_x\geq \varepsilon/4t$ for every $x\in L$; in particular, all $c_x$ are nonnegative.
Finally, we set $ {\cal P}'$ to be ${\cal C}$ divided by its 1-norm $\|{\cal C}\|_1=\sum |c_x|$, so that ${\cal P}'$ is a proper probability distribution. Since ${\cal P}'$ is $|L|$-sparse,  computing ${\cal P}'$ from ${\cal C}$ can be done in $O(|L|) = $poly$(t, 1/\varepsilon)$ time. Putting everything together, the total runtime for computing ${\cal P}'$ scales as poly$(k, t, \frac{1}{\varepsilon}, \log \frac{1}{\delta})$. We now show that ${\cal P}'$ is also $O(\varepsilon)$-close to ${\cal P}$. The argument is straightforward and fully analogous to the one in \Cref{sect_sparse_basic}, cf. (\ref{eq:bound_P_t})-(\ref{eq:argument_normalize}). Since $\|{\cal C}- {\cal P}\|_1\leq 3\varepsilon$ and $\|{\cal P}\|_1=1$ we have \be\label{bound_tilde_P} 1-3\varepsilon\leq \|  {\cal C}\|_1\leq 1+3\varepsilon.\ee We then find \be \|{\cal P}' - {\cal P}\|_1 &=& \frac{\| {\cal C}- \|{\cal C}\|_1 \cdot {\cal P}\|_1}{\|{\cal C}\|_1}\leq \frac{\| {\cal C}- \|{\cal C}\|_1 \cdot {\cal P}\|_1}{1-3\varepsilon}\nonumber\\ &\leq& \frac{\| {\cal C}- {\cal P}\|_1}{1-3\varepsilon} + \frac{|1- \|{\cal C}\|_1|\cdot \|{\cal P}\|_1}{1-3\varepsilon} \leq \frac{6\varepsilon}{1-3\varepsilon}.\ee
Then, for $\varepsilon \leq 1/6$, we have $\|{\cal P}' - {\cal P}\|_1\leq 12\varepsilon$. Note also that $p_x'\geq \varepsilon/8t$ for all $x\in L$ follows by combining the inequalities $c_x\geq \varepsilon/4t$ and $\|{\cal C}\|\leq 1+ 3\varepsilon$ and $\varepsilon\leq 1/6$.

Finally, we show how to sample ${\cal P}'$. For a bit string $x_1\cdots x_m$ with $m$ between 1 and $k$, let $p'(x_1\cdots x_m)$ denote the marginal probability of ${\cal P}'$ for obtaining $x_1\cdots x_m$ on the first $m$ bits.  Since ${\cal P}$ is $s$-sparse with $s= O(t/\varepsilon)$, each $p'(x_1\cdots x_m)$  can be computed from ${\cal P}'$ in poly$(s) = $ poly$(t, 1/\varepsilon)$ time on input of $x_1\cdots x_m$. By a standard argument, the property that all such marginals can be computed, allows to sample ${\cal P}'$  in poly$(k, t, 1/\varepsilon)$ time \cite{JVV86, valiant02, TD04}.
\end{proof}

\section{Classical simulation of CT states}\label{sect_CT}
Here we review two classical simulation results for CT states which will be used in the proofs of our results.  An $n$-qubit unitary operator $U$ is said to be efficiently computable basis-preserving if there exist efficiently computable functions $f, f':B_n\to B_n$ and $g, g' :B_n\to \mathbbm{C}$ where $|g(x)|=1 =|g'(x)|$ for all $x\in B_n$, such that, for every computational basis state $|x\rangle$, one has \be U|x\rangle = g(x)|f(x)\rangle\quad\mbox{and}\quad U^{\dagger}|x\rangle = g'(x)|f'(x)\rangle\ee A notable example of efficiently computable basis preserving operations is given by operators comprising tensor products of Pauli matrices $\id,X,Y,Z$.
\begin{lemma}[\cite{vdN11}] \label{thm:overlap}
Let $\ket{\psi}$ and $\ket{\varphi}$ be CT $n$-qubit states and let $A$ be an efficiently computable basis-preserving $n$-qubit operation. Then there exists a randomized classical algorithm with runtime poly$(n,1/\varepsilon, \log\frac{1}{\delta})$ which outputs an approximation of $\bra{\psi}A\ket{\varphi}$ with accuracy $\varepsilon$ and success probability at least $1-\delta$.
\end{lemma}

\begin{lemma}[\cite{vdN11}] \label{thm:partial overlap}
Let $\ket{\psi}$ and $\ket{\varphi}$ be CT $n$-qubit states, let $\ket{\xi}$ and $\ket{\chi}$ be CT $k$-qubit states with $k \leq n$. Then there exists a randomized classical algorithm with runtime poly$(n,1/\varepsilon, \log\frac{1}{\delta})$ which outputs an approximation of $\bra{\varphi}[\ket{\xi}\bra{\chi} \otimes \id]\ket{\psi}$ with accuracy $\varepsilon$ and success probability at least $1-\delta$.
\end{lemma}
The above results are slightly more detailed then the corresponding results in \cite{vdN11} since the latter reference does not provide explicit information about the scaling with $\varepsilon$ and $\delta$. For completeness, proofs of \Cref{thm:overlap} and \Cref{thm:partial overlap} (which are straightforward extensions of the proofs in \cite{vdN11}) are given in \Cref{sec:sharpCT}.

\section{Proofs of main results}\label{sect_proofs}

\subsection{Proof of \Cref{thm_main1}}

The proof will be obtained by showing that the output distribution of any quantum circuit considered in \Cref{thm_main1} satisfies the conditions of \Cref{thm_central}. We introduce some further basic definitions. For any positive integer $d$, let $X_{d}$, $Z_{d}$ be \emph{generalized Pauli operators} (also known as \emph{Weyl operators}) \cite{gottesman99}, which act on the $d$-level computational basis states $\ket{x}$ (with $x\in\mathbbm{Z}_d$) as follows
\begin{align}
X_d \ket{x} &= \ket{x+1}\\
Z_d \ket{x} &=e^{\frac{2 \pi i}{d}x} \ket{x}
\end{align}
where $x+1$ is defined modulo $d$. Note that the order of both $X_d$ is $d$ (i.e. is the smallest integer $r\geq 2$ satisfying $X_d^r = I$ is precisely $d$), as is the order of $Z_d$. Let ${\cal F}_d$ denote the Fourier transform over $\mathbb{Z}_d$. A straightforward application of definitions \cite{gottesman99} shows that
\be\label{Fourier_pauli}
\mathcal{F}_d^{\dagger} Z_d \mathcal{F}_d = X_d.
\quad\mbox{and}\quad\mathcal{F}_d Z_d \mathcal{F}_d^{\dagger} = X_d^{\dagger}.
\ee
\Cref{thm_main1} now follows immediately from \Cref{thm_central} in combination with the following result:
\begin{lemma}\label{thm_firstprop}
Let ${\cal P}$ be a probability distribution on $B_k$ arising from a quantum circuit satisfying conditions (a)-(b) in \Cref{thm_main1}. Let ${\cal P}_m$ denote the marginal distributions arising from measurement of the first $m$ qubits, for $m=1, \dots, k$ (with ${\cal P}\equiv {\cal P}_m$). Then each ${\cal P}_m$ is additively approximable.
\end{lemma}
\begin{proof}
Without loss of generality we let $S$ be the set of first  $k$ qubits. For a $k$-bit string $x=(x_1, \dots, x_k)$, consider the associated $k$-bit integer  $\hat{x}:= x_12^0 + x_2 2 + \cdots + x_k 2^{k-1}$. The standard basis states of a $k$-qubit system will be labeled both by the set of $k$-bit strings $x$ and the associated integers $\hat{x}$ depending on which formulation is most convenient. Below we will use the basic fact that, for any $m=1, \dots, k$,  \be\label{mod_2_to_the_m} \hat{x}\mod 2^m = x_12^0 +  \cdots + x_m 2^{m-1}.\ee
Let $m\in\{1, \dots,  k\}$. For an $m$-bit string $y=y_{1} \cdots y_m$, consider the projector (acting on $k$ qubits) \be  |y_{1}\cdots y_m\rangle\langle y_{1}\cdots y_m|\otimes I\equiv P(y)\ee where $I$ denotes the identity on the last $k-m$ qubits. Thus $P(y)$ is the projector onto those $k$-qubit computational basis states $|x\rangle$ where the first $m$ bits of $x$ coincide with  $y$.  Owing to (\ref{mod_2_to_the_m}), this means that $P(y)$ is the projector on those computational $|x\rangle$ satisfying $\hat{x} $ mod $2^m = \hat{y}$, where $\hat{y}:=y_12^0 +  \cdots + y_m 2^{m-1}$.  Let $Z_{2^k}\equiv Z$ and $X_{2^k}\equiv X$ denote the generalized Pauli operators acting on $\mathbbm{C}^{2^k}$. A straightforward application of the definition of $Z$ shows that  \be \hat{x} \mod 2^m = \hat{y} \quad \mbox{iff} \quad \alpha^{\hat{y}}Z^{2^{k-m}}|\hat{x}\rangle = |\hat{x}\rangle \quad \mbox{ with }\alpha:= e^{-\frac{2\pi i}{2^m}}.\ee This implies that $P(y)$ coincides with the projector onto the eigenspace of $M:=\alpha^{\hat{y}}Z^{2^{k-m}}$ with eigenvalue $1$. This projector can be obtained by averaging over all powers of $M$; since the order of $M$ is $2^m$ (recall that the order of $Z$ is $2^k$), this implies that \be P(y) = \frac{1}{2^m}\sum_{u=0}^{2^m -1} M^u.\ee Let ${\cal F}\equiv {\cal F}_{2^k}$ denote the Fourier transform modulo $2^k$. We consider the scenario where ${\cal F}$ is applied in the block $U_2$; the case where ${\cal F}^{\dagger}$ is applied is treated in full analogy and is omitted here. Denoting $N:=\alpha^{\hat{y}}X^{2^{k-m}}$ (i.e. we replace $Z$ by $X\equiv X_{2^k}$ in the definition of $M$) and recalling the first identity of \cref{Fourier_pauli} we find \be\label{sum_N} {\cal F}^{\dagger}P(y){\cal F} = \frac{1}{2^m}\sum_{u=0}^{2^m -1} N^u.\ee Now denote the $n$-qubit CT state generated after application of the block $U_1$ by $|\mbox{CT}\rangle$. Furthermore denote the marginal probability of obtaining the bit string $y$ when measuring the first $m$ qubits at the end of the circuit by $p(y)$. Then \be p(y) = \langle \mbox{CT}| [{\cal F}^{\dagger}P(y) {\cal F}]\otimes I|\mbox{CT}\rangle\ee where $I$ denotes the identity acting on the last $n-k$ qubits. Using \Cref{sum_N} we find \be p(y_1\cdots y_m) = \frac{1}{2^m}\sum_{u=0}^{2^m -1} \langle \mbox{CT}|N^u\otimes I|\mbox{CT}\rangle.\ee It easily follows from the definition of $N$ that each $N^{u}\otimes I$ is efficiently computable basis-preserving (as defined in section \ref{sect_CT}). Together with \Cref{thm:overlap} this implies that the function $u\in\mathbbm{Z}_{2^m}\to \langle\mbox{CT}|N^u\otimes I|\mbox{CT}\rangle$ is additively approximable. But then \Cref{thm_approximate_F} implies that $y\to p(y)$ is additively approximable as well.
\end{proof}

\subsection{Proof of \Cref{thm_main2}}

Similar to the proof of \Cref{thm_main1}, also the proof of \Cref{thm_main2} follows immediately by showing that the output distribution of any quantum circuit considered in \Cref{thm_main2} satisfies the conditions of \Cref{thm_central}. The latter is done next.
\begin{lemma}\label{thm_secondprop}
Let ${\cal P}$ be a probability distribution on $B_k$ arising from a quantum circuit satisfying conditions (a)-(b') in \Cref{thm_main2}. Let ${\cal P}_m$ denote the marginal distributions arising from measurement of the first $m$ qubits, for $m=1, \dots, k$ (with ${\cal P}\equiv {\cal P}_m$). Then each ${\cal P}_m$ is additively approximable.
\end{lemma}
\begin{proof}
We prove the result for qubit systems; the proof will carry over straightforwardly to  systems of qudits of potentially different dimensions. Without loss of generality we let $S$ be the set of first  $k$ qubits. For an $m$-bit string $y=y_{1} \cdots y_m$ with $m\leq k$,  let $p(y)$ denote the marginal probability of the outcome $y_1\cdots y_m$ when measuring the first $m$ qubits at the end of the circuit. We need to show that the function $y\to p(y)$ is additively approximable. Denote the CT state generated after application of the block $U_1$ by $|\mbox{CT}\rangle$. Since $U_2=u_1\otimes \cdots \otimes u_n$ is a tensor product operator and since $|y\rangle$ is a product state, we have  \be p(y)= \bra{\mbox{CT}} U^{\dagger}[\ket{y}\bra{y}  \otimes \id] U \ket{{\mbox{CT}}} = \bra{\mbox{CT}}\ket{\alpha}\bra{\alpha}  \otimes \id \ket{{\mbox{CT}}}\ee for some $m$-qubit tensor product state $|\alpha\rangle$ (with efficiently computable description). Since product states are CT, \Cref{thm:partial overlap} immediately implies that $y\to p(y)$ is additively approximable.
\end{proof}
\subsection{Proof of \Cref{thm_main3} and \Cref{thm_main4}}

\begin{lemma}\label{thm_CT_chernoff}
Let $|\mbox{CT}\rangle$ be an $n$-qubit CT state, let $U=U_1\otimes \cdots\otimes U_n$ be a unitary tensor product operator and let ${\cal F}$ denote the Fourier transform modulo $2^n$. Then the following functions are additively approximable (where $x=x_1\cdots x_n$ is an $n$-bit string):\be x&\to& \langle x|{\cal F}|\mbox{CT}\rangle\\x&\to& \langle x|{\cal F}^{\dagger}|\mbox{CT}\rangle\\ x&\to& \langle x|U|\mbox{CT}\rangle.\ee
The last function is still additively approximable when generalized to tensor product operators acting on $n$ qudit systems with potentially different dimensions.
\end{lemma}
\begin{proof}
A straightforward application of definitions shows that the states ${\cal F}|x\rangle$, ${\cal F}^{\dagger}|x\rangle$ and $U|x\rangle$ are CT. The result then immediately follows from \Cref{thm:overlap} (with $A$ being the identity).
\end{proof}
\begin{lemma}\label{thm_phase_estimate}
Let $c, c'$ be two complex numbers satisfying $c\neq 0$ and $|c- c'|\leq \alpha$ for some $\alpha>0$. Let $c= \theta |c|$ where $\theta$ is the phase of $c$ and similarly $c' = \theta' |c'|$. Then $|\theta - \theta'|\leq 2\alpha/|c|$.
\end{lemma}
\begin{proof}
Since $|c- c'|\leq \alpha$, we have $||c|- |c'||\leq \alpha$. Then \be |\theta - \theta'| |c| = |c - \theta'|c||\leq |c-c'|+ |c' - \theta'|c|| = |c-c'|+ ||c'| - |c||\leq 2\alpha.\ee
\end{proof}
Next we prove \Cref{thm_main3} and \Cref{thm_main4}.  Let $|\psi_{\mbox{\scriptsize{out}}}\rangle$  denote the final state in any of the settings considered in \Cref{thm_main3} and \Cref{thm_main4}. We write $\langle x|\psi_{\mbox{\scriptsize{out}}}\rangle = \gamma_x\sqrt{p_x}$ where $\gamma_x$ is the phase and $p_x$ the modulus squared, so that ${\cal P}=\{p_x\}$ is the probability distribution arising from measuring all qubits of $|\psi_{\mbox{\scriptsize{out}}}\rangle$ in the computational basis.
Since $|\psi_{\mbox{\scriptsize{out}}}\rangle$ is $\sqrt{\varepsilon}$-approximately $t$-sparse, ${\cal P}$ is $\varepsilon$-approximately $t$-sparse by \Cref{thm_sparse_equiv}.  Recalling \Cref{thm_firstprop} and \Cref{thm_secondprop}, we find that all conditions of \Cref{thm_central} are fulfilled. Thus there exists a randomized classical algorithm with runtime poly$(n, t, \frac{1}{\varepsilon}, \log \frac{1}{\delta})$ which outputs  an $s$-sparse probability distribution ${\cal P}'=\{p_x'\}$ where $s=O(t/\varepsilon)$  such that, with probability at least $1-\delta$,  $\|{\cal P}'- {\cal P}\|_1\leq 12\varepsilon$. Let $L$ be the list of bit strings as in the proof of \Cref{thm_central}. Recall from the latter proof also the following properties: $|L|\leq 2t/\varepsilon$; $L$ is precisely the support of ${\cal P}'$; $p_x\geq \varepsilon/2t$ for every $x\in L$.

Thus far we have  computed an approximation ${\cal P}'$ of the probability distribution ${\cal P}$. Next we will also approximately compute the amplitudes of $ |\psi_{\mbox{\scriptsize{out}}}\rangle$ by employing \Cref{thm_CT_chernoff}.  For every $x\in L$ we compute a complex number $a_x$ satisfying
\be
|a_x - \langle x|\psi_{\mbox{\scriptsize{out}}}\rangle|\leq \sqrt{\varepsilon^3/8t}.
\ee
Owing to \Cref{thm_CT_chernoff},  the function $x\to \langle x|\psi_{\mbox{\scriptsize{out}}}\rangle$ is additively approximable. Therefore each individual $a_x$ can be computed with success probability at least $1- \delta/|L|$ in time $N=$ poly$(n, t, \frac{1}{\varepsilon}, \log \frac{1}{\delta})$. Thus the total runtime for computing all $a_x$ is $|L|T = $ poly$(n, t, \frac{1}{\varepsilon}, \log \frac{1}{\delta})$ and the total success probability is at least $1-\delta$. We then compute the complex phase $\theta_x$ of each $a_x$ (which requires $O(|L|)$ computational steps in total) and define the  state \be |\varphi\rangle:= \sum_{x\in L} \theta_x \sqrt{p_x'}|x\rangle. \ee Note that $|\varphi\rangle$ has 2-norm equal to 1: indeed $\||\varphi\rangle\|_2^2$ coincides with $\sum_{x\in L } p_x'$ which equals 1 since $L$ coincides with the support of ${\cal P}'$. Next we prove that $|\varphi\rangle$ is $O(\sqrt{\varepsilon})$-close to $|\psi_{\mbox{\scriptsize{out}}}\rangle$. The idea of the argument is rather straightforward but the details will be somewhat tedious.

First we show that the phase $\theta_x$ is close to $\gamma_x$ for every $x\in L$ (recall that the latter is the phase of $\langle x|\psi_{\mbox{\scriptsize{out}}}\rangle$): using \Cref{thm_phase_estimate} and recalling that $p_x\geq \varepsilon/2t$, we have \be|\theta_x - \gamma_x|\leq 2\cdot \sqrt{\frac{ \varepsilon^3}{8t}} \cdot\frac{1}{\sqrt{p_x}} \leq \varepsilon. \ee
This implies that \be
\| \sum_{x\in L} (\theta_x - \gamma_x) \sqrt{p_x'}|x\rangle \|_2^2= \sum_{x\in L} |\theta_x - \gamma_x|_2^2 p_x'\leq \varepsilon^2 \sum_{x\in L}p_x' \leq \varepsilon^2.
\ee
For every two numbers $a, b\geq 0$ we have $|a-b|^2\leq |a^2-b^2|$. This implies that \be \sum |\sqrt{p_x'} - \sqrt{p_x}|^2 \leq \sum |p_x' - p_x| = \|{\cal P}'- {\cal P}\|_1\leq 12\varepsilon\ee where the sums are over all $x\in B_n$. Hence \be \||\psi_{\mbox{\scriptsize{out}}}\rangle - \sum_{x\in L} \gamma_x \sqrt{p_x'}|x\rangle\|_2^2 &=& \sum_{x\in L} |\gamma_x\sqrt{p_x}- \gamma_x\sqrt{p_x'}|^2 + \sum_{x\notin L} p_x\nonumber\\ &=&\sum_{x\in L} |\sqrt{p_x}- \sqrt{p_x'}|^2 + \sum_{x\notin L} p_x\nonumber\\ &=& \sum_{x\in B_n} |\sqrt{p_x} - \sqrt{p_x}'|^2\leq 12\varepsilon\ee where in the last equality we used that $p_x'=0$ for all $x\notin L$.
Writing \be |\varphi\rangle = \sum_{x\in L} \gamma_x \sqrt{p_x'}|x\rangle +  \sum_{x\in L} (\theta_x - \gamma_x) \sqrt{p_x'}|x\rangle \ee and using the triangle inequality, we then find \be \||\psi_{\mbox{\scriptsize{out}}}\rangle - |\varphi\rangle \|_2 &\leq& \||\psi_{\mbox{\scriptsize{out}}}\rangle - \sum_{x\in L} \gamma_x \sqrt{p_x'}|x\rangle\|_2 + \| \sum_{x\in L} (\theta_x - \gamma_x) \sqrt{p_x'}|x\rangle \|_2\nonumber\\ &\leq& \sqrt{12\varepsilon} + \varepsilon \leq 5\sqrt{\varepsilon}.\ee

\subsection{Proof of \Cref{thm_main5}}

Denote by ${\cal P}=\{p_x: x\in B_n\}$ the probability distribution arising from a standard basis measurement of all $n$ qubits performed on the state ${\cal F}_{2^n}^{\dagger}|\psi\rangle$. Then  $p_x=|\hat\psi_x|^2$. It follows from \Cref{thm_firstprop} that ${\cal P}$ and its marginals ${\cal P}_m$ fulfill all conditions of \Cref{thm_compute_large_coeff}. The latter result then immediately implies the existence of a classical algorithm with runtime poly$(k, \frac{1}{\theta}, \log\frac{1}{\pi})$ which outputs a list $L =\{x^1, \dots, x^l\}$ where $l\leq 2/\theta$ such that, with probability at least $1-\pi$, conditions (a) and (b) in \Cref{thm_main5} are fulfilled.  Furthermore, \Cref{thm_CT_chernoff} implies that, given any $x\in B_n$, there exists a classical algorithm with runtime poly$(n, 1/\varepsilon, \log{\frac{1}{\delta}})$ which, with probability at least $1-\delta$, outputs an $\varepsilon$-approximation of $\hat\psi_x$, since $\hat\psi_x = \langle x|{\cal F}_{2^n}^{\dagger}|\psi\rangle$.

Fully analogously, for $U=U_1\otimes \cdots\otimes U_n$ let ${\cal P}=\{p_x\}$ be the probability distribution arising from a standard basis measurement of all $n$ qubits performed on the state $U^{\dagger}|\psi\rangle$. The extension of \Cref{thm_main5} to the product basis $\{U|x\rangle\}$  is now obtained by combining \Cref{thm_secondprop}, \Cref{thm_compute_large_coeff}, and  \Cref{thm_CT_chernoff}.

\section{Further research}

In the classical simulation algorithms given in this paper, we have not optimized the degree or constants involved in the polynomial-time simulation. While our algorithm is a generalization of \cite{KM91,GL89}, for optimal performance one could try to adapt the more advanced, query-optimal algorithm of \cite{HIKP12b} to our setting.

\bibliography{quantum}
\bibliographystyle{alpha}

\appendix

\section{Proof of lemma \ref{thm:chernoff}}\label{sec:app_chernoff}

We recall the standard Chernoff-Hoeffding bound for real-valued random variables.
\begin{theorem}[Chernoff-Hoeffding bound] \label{thm:hoeffding}
Let $X_1,\dots, X_T$ be i.i.d. real random variables. Assume that $|X_i|\leq 1$ and denote $E:= \mathbf{E}X_i$. Then \be\label{hoeff} \mbox{Prob} \left\{ \left|\frac{1}{T}\sum_{i=1}^T X_i - E\right| \leq \varepsilon\right\} \geq 1-2 e^{- \frac{T\varepsilon^2}{2}}.\ee
\end{theorem}
\noindent The proof of the complex-valued version of the Chernoff-Hoeffding bound as given in lemma \ref{thm:chernoff} is an immediate corollary of the real-valued version, as follows. For complex-valued random variables $X_1, \dots, X_T$ we apply \Cref{thm:hoeffding} independently to the real and imaginary parts of the $X_i$, where we choose $\tilde{\varepsilon}=\frac{\varepsilon}{\sqrt{2}}$. Denoting $Y:=\frac{1}{T}\sum_{i=1}^{T} X_i - E$, this yields lower bounds for the probabilities that $Re(Y) \leq \tilde{\varepsilon}$ and $Im(Y) \leq \tilde{\varepsilon}$.  Putting things together we find \be \mbox{Prob} \left\{ \left|\frac{1}{T}\sum_{i=1}^T X_i - E\right| \leq \varepsilon\right\} \geq 1-4 e^{- \frac{T\varepsilon^2}{4}}.\ee

\section{Proofs of lemmas \ref{thm:overlap} and \ref{thm:partial overlap}} \label{sec:sharpCT}

In this section we give explicit quantitative versions of the definition and theorems about CT states, which were only stated implicitly in \cite{vdN10}.

\begin{definition}[Computationally Tractable (CT) states] \label{def:sct}
An $n$-qubit state $\ket{\psi}$ is called `computationally tractable' (CT) if the following conditions hold:
\begin{enumerate}
\item \ [Sample] it is possible to sample in time $s_{\ket{\psi}}=O(poly(n))$ with classical means from the probability distribution $Prob(x) = |\bracket{x}{\psi}|^2$ on the set of $n$-bit strings $x$. \label{def:sct:sample}
\item \ [Query] upon input of any bit string $x$, the coefficient $\bracket{x}{\psi}$ can be computed in $c_{\ket{\psi}}=O(poly(n))$ time on a classical computer. \label{def:sct:query}
\end{enumerate}
\end{definition}
The proof of lemma \ref{thm:overlap} will follow immediately from the following result:
\begin{lemma} Let $\ket{\psi}$ and $\ket{\varphi}$ be two CT $n$-qubit states and let $s=s_{\ket{\psi}}+s_{\ket{\varphi}}$, $c=c_{\ket{\psi}}+c_{\ket{\varphi}}$. Then there exists a randomized classical algorithm to compute $\mu$ such that $|\bracket{\varphi}{\psi} - \mu| \leq \varepsilon$
in time $O(\frac{s+c}{\varepsilon^2} \log(\frac{4}{\delta}))$ with error probability $\delta$.
\end{lemma}
\begin{proof}
Denote $p_x := |\bracket{x}{\psi}|^2$ and $q_x := |\bracket{x}{\varphi}|^2$ . Since $\ket{\psi}$ and $\ket{\varphi}$ are CT states, it is possible to sample from the probability distributions $\{p_x\}$ and $\{q_x\}$ in time $s$ (\Cref{def:sct}, \Cref{def:sct:sample}). Define the function
$\alpha: \{0,1\}^n \mapsto \{0,1\}$ by $\alpha(x) = 1$ if $p_x \geq q_x$ and $\alpha(x)=0$ otherwise, for every $n$-bit string $x$, and define the function $\beta$ by $\beta(x):=1-\alpha(x)$. Then $\alpha$ and $\beta$ can be computed in time $O(c)$ since $p_x$ and $q_x$ can be computed in time $c$ each by \Cref{def:sct:query} in \Cref{def:sct}. The overlap $\bracket{\varphi}{\psi}$ is equal to
\begin{equation}
\bracket{\varphi}{\psi} = \sum \bracket{\varphi}{x}\bracket{x}{\psi}\alpha(x) + \sum \bracket{\varphi}{x}\bracket{x}{\psi}\beta(x)
\end{equation}
where the sums are over all $n$-bit strings $x$. Defining the functions $F$ and $G$ by
\begin{equation}
F(x) = \frac{\bracket{\varphi}{x}\bracket{x}{\psi}}{p_x}\alpha(x), \;\;\;
G(x) = \frac{\bracket{\varphi}{x}\bracket{x}{\psi}}{q_x}\beta(x)
\end{equation}
we have $\bracket{\varphi}{\psi}=\langle F \rangle + \langle G \rangle$, where $\langle F \rangle=\sum p_x F(x)$ and $\langle G \rangle=\sum p_x G(x)$.
It follows from the query property (\Cref{def:sct}, \Cref{def:sct:query}) of CT states, that $F$ and $G$ can be evaluated in time $O(c)$.
Furthermore, both $|F(x)|$ and $|G(x)|$ are not greater than $1$. It thus follows from \Cref{thm:chernoff}, that both $\langle F \rangle$ and $\langle G \rangle$ can be approximated with accuracy $\varepsilon/2$ and error probability at most $\delta/2$ by estimating the averages over samples from the distributions $p_x$ and $q_x$, respectively. More precisely, let $X_i$, $1 \leq i \leq T$, be samples drawn from distribution $\{p_x\}$  with $T=\frac{16}{\varepsilon^2} \log(\frac{8}{\delta})$, and let $\mu_F=\frac{1}{T}\sum_{i=1}^{T} F(X_i)$, (and similarly for samples $Y_i$ drawn from $\{q_x\}$, $\mu_G=\frac{1}{T}\sum_{i=1}^{T} G(Y_i)$), then it follows from \Cref{thm:chernoff} that
\begin{eqnarray}
\Pr\left\{\left|\mu_F - \langle F \rangle \right| \leq \varepsilon/2 \right\} \geq 1 - \delta/2 \\
\Pr\left\{\left|\mu_G - \langle G \rangle \right| \leq \varepsilon/2 \right\} \geq 1 - \delta/2
\end{eqnarray}
Thus we conclude that $\bracket{\varphi}{\psi}$ can be approximated by $\mu=\mu_F+\mu_G$ in time $O(\frac{s+c}{\varepsilon^2} \log(\frac{4}{\delta}))$ such that
\begin{equation}
\Pr\left\{\left|\mu - \bracket{\varphi}{\psi} \right| \leq \varepsilon \right\} \geq 1 - \delta
\end{equation}
\end{proof}

The proof of lemma \ref{thm:partial overlap} is obtained by noting that any partial overlap of $n$-qubit CT states (as considered in lemma \ref{thm:partial overlap}) can be re-expressed (via a poly$(n)$ time classical reduction) as a complete overlap $\langle\phi|\phi'\rangle$ where $|\phi\rangle$ and $|\phi'\rangle$ are CT states on $O(n)$ qubits. Invoking lemma \ref{thm:overlap} then proves the result.

\end{document}